\documentclass[number]{ReportTemplate}
\usepackage{amssymb}
\usepackage{amsfonts}
\usepackage{amsmath}
\usepackage{mathtools}
\usepackage{amsthm}
\usepackage{bm}
\usepackage{graphicx}
\usepackage{algorithm}
\usepackage{algorithmic}
\usepackage[american]{babel}
\usepackage{subfigure}
\usepackage{hyperref}
\usepackage{caption,tabularx,booktabs}
\usepackage{multirow}
\mathtoolsset{showonlyrefs,showmanualtags}

\usepackage[numbers]{natbib}
\usepackage{makecell}
\usepackage{xcolor}
\usepackage{amsmath}\allowdisplaybreaks

\newdefinition{rmk}{Remark}
\newcommand{\bsp}{\{0,1\}^n}

\newcommand{\expect}[1]{\mathrm{E}(#1)}
\newcommand{\pr}[1]{\mathrm{P}(#1)}
\newcommand{\pmut}{\mathrm{P}_\mathrm{mut}}
\newcommand{\pp}{\mathrm{P}}
\newcommand{\blue}[1]{{#1}}
\newcommand{\blueBl}{{\Big(}}
\newcommand{\blueBr}{{\Big)}}

\begin{document}
\let\WriteBookmarks\relax
\def\floatpagepagefraction{1}
\def\textpagefraction{.001}
%\shorttitle{Running Time Analysis of the (1+1)-EA for Robust Linear Optimization}
%\shortauthors{Bian et~al.}
%\begin{frontmatter}

\title{Running Time Analysis of the (1+1)-EA for Robust Linear Optimization*}                     
%\tnotemark[1,2]

%\tnotetext[1]{This document is the results of the research   project funded by the National Science Foundation.}

%\tnotetext[2]{The second title footnote which is a longer text matter  to fill through the whole text width and overflow into another line in the footnotes area of the first page.}

\author[1]{Chao Bian}                       
\ead{chaobian12@gmail.com}
\address[1]{State Key Laboratory for Novel Software Technology, Nanjing University, Nanjing 210023, China}

\author[1]{Chao Qian}
\ead{qianc@lamda.nju.edu.cn}
\cortext[cor1]{ Chao Qian is the corresponding author.}

\author[2]{Ke Tang}
\ead{tangk3@sustech.edu.cn}
\address[2]{Shenzhen Key Laboratory of Computational Intelligence, Department of Computer Science and Engineering, Southern University of Science and Technology, Shenzhen 518055, China}

\author[1]{Yang Yu}                      
\ead{yuy@lamda.nju.edu.cn}

\begin{abstract}
Evolutionary algorithms (EAs) have found many successful real-world applications, where the optimization problems are often subject to a wide range of uncertainties. To understand the practical behaviors of EAs theoretically, there are a series of efforts devoted to analyzing the running time of EAs for optimization under uncertainties. Existing studies mainly focus on noisy and dynamic optimization, while another common type of uncertain optimization, i.e., robust optimization, has been rarely touched. In this paper, we analyze the expected running time of the (1+1)-EA solving robust linear optimization problems (i.e., linear problems under robust scenarios) with a cardinality constraint $k$. Two common robust scenarios, i.e., deletion-robust and worst-case, are considered. Particularly, we derive tight ranges of the robust parameter $d$ or budget $k$ allowing the (1+1)-EA to find an optimal solution in polynomial running time, which disclose the potential of EAs for robust optimization.
\end{abstract}
%
%\begin{keywords}
%evolutionary algorithms \sep running time analysis \sep robust optimization \sep  computational complexity
%\end{keywords}

\maketitle

\section{Introduction}

Evolutionary algorithms (EAs)~\cite{back:96} are general-purpose heuristic optimization algorithms, and have been widely applied to solve real-world optimization problems~\cite{liang-fcs20}, which are often subject to various uncertainties. Meanwhile, theoretical analysis, particularly running time analysis of EAs has achieved a lot of progress~\cite{neumann2010bioinspired,auger2011theory} during the last two decades. Though most of the existing theoretical studies focus on exact optimization, uncertain optimization has attracted much attention recently~\cite{giessen2014robustness,qian2018noise,shi-algo18-dynamic}.

Generally, optimization under uncertainties can be classified into three categories~\cite{jin2005evolutionary}: noisy optimization, dynamic optimization, and robust optimization.\footnote{The category ``fitness approximation" and the category ``noise" in~\cite{jin2005evolutionary} are similar. Both of them introduce errors into fitness evaluation. The main difference is that the former introduces deterministic error whereas the latter introduces random error. Thus, we combined them together under the umbrella of ``noise". } For noisy optimization, one cannot obtain an exact objective function value, but only a noisy one. The classic (1+1)-EA was first studied on the OneMax and LeadingOnes problems under various noise models~\cite{giessen2014robustness,qian2018noise,droste2004analysis,qian2018ppsn}. The studies show that the (1+1)-EA is efficient, i.e., it can find an optimal solution in polynomial running time, only under low noise levels. Later studies mainly proved the robustness of different strategies against noise, including using populations~\cite{giessen2014robustness,dang2015efficient,doerr2018gecco,prugel2015run,dirk2018gecco}, sampling~\cite{qian2018noise,qian2016sampling} and threshold selection~\cite{qian2015noise}. There is also a sequence of papers analyzing the running time of the compact genetic algorithm~\cite{friedrich2015benefit} and a simple ant colony optimization algorithm~\cite{friedrich2015robustness,sudholt2012simple,doerr2012ants,feldmann2013optimizing} solving noisy problems, including OneMax as well as the combinatorial optimization problem single destination shortest paths.

For dynamic optimization, the objective function or the constraints of the problem to be solved may change over time, and thus, the optimal solutions may change over time. The goal of the optimizer is to track the optimal solutions continuously. Droste~\cite{droste-cec02-dynamic-om} first analyzed the dynamic OneMax problem, where the fitness of a solution is the number of bits matching a target bitstring. He proved that the expected running time of the (1+1)-EA is polynomial when the target bitstring changes one uniformly chosen bit with probability $O(\log n/n)$ in each iteration, where $n$ is the problem size. K\"otzing et al.~\cite{kotzing-foga15-dynamic} re-proved some results in~\cite{droste-cec02-dynamic-om} and investigated the extended dynamic OneMax problem. Shi et al.~\cite{shi-algo18-dynamic} considered the linear pseudo-Boolean functions under dynamic uniform constraints. Dynamic versions of some combinatorial optimization problems have also been studied, including shortest paths~\cite{lissovoi-tcs15-aco-ssp}, vertex cover~\cite{pourhassan-gecco15-dynamic-vc,pourhassan-ssci17-dynamic-vc}, weighted vertex cover~\cite{shi-gecco18-dynamic-vc} and makespan scheduling~\cite{neumann-ijcai15-dynamic-makespan}.

For robust optimization, the design variables or the environmental parameters may change after obtaining desired solutions, which thus have to be robust against these changes. Previous studies of EAs on robust optimization are mainly empirical, e.g.,~\cite{deb2006introducing,beyer2007robust,fu2015robust,zhou-scis19}. To the best of our knowledge, running time analysis has been rarely touched.

In fact, robust pseudo-Boolean optimization, also called robust subset selection, has been theoretically studied in the Machine Learning community. The subset selection problem is to select a subset of size at most $k$ from a total set $V\!\!=\!\{\!v_1,v_2,\dots,v_n\}$ of $n$ items  for maximizing some given objective function. The problem can be formally described~as
\begin{align}\label{form-subset}
\max_{X\subseteq V}f(X) \quad \text{s.t.} \quad |X|\le k,
\end{align}
where $|\cdot|$ denotes the size of a set. The subset selection problem has many applications. For example, in the sparse regression problem~\cite{tropp-tit04}, one needs to select a subset of observation variables to best approximate the predictor variable; in the influence maximization problem~\cite{kempe2003maximizing}, one needs to select a subset of users from a social network to maximize its influence spread; in the sensor placement problem~\cite{krause-jmlr08-sensor}, one needs to select a few places to install sensors such that the information gathered is maximized. For pseudo-Boolean optimization, a solution $x\in\bsp$ is a Boolean vector, which naturally characterizes a subset. That is, the $i$-th item is selected if and only if $x_i=1$. Meanwhile, the selection of at most $k$ items is actually the cardinality constraint for $x$, i.e., $|x|_1\le k$, where $|\cdot|_1$ denotes the number of 1-bits. {In this paper, a vector $x\in\bsp$} and its corresponding subset will not be distinguished for notational convenience.
%Thus, the subset selection problem can be represented as a pseudo-Boolean optimization problem
%\begin{align}\label{form-pseudo}
%\max_{x\in\bsp}f(x) \quad \text{s.t.} \quad |x|_1\le k,
%\end{align}

For robust pseudo-Boolean optimization, Krause et al.~\cite{krause-jmlr08-rsos} first considered the deletion-robust setting. In many applications, one requires robustness in the solution such that the objective value degrades as little as possible when some items in the solution are deleted. For example, in the problem of sensor placement for monitoring spatial phenomena, the goal is to select a few locations to install sensors to maximize the coverage; but some sensors may fail, and it is desired that the remaining sensors have good coverage~\cite{orlin-ipco16-deletion}. In the problem of influence maximization, the goal is to spread the word of a new product by targeting the most influential users; but some of the users from the targeted set might refuse to spread the word, and we still want to maximize the information spread by the remaining users~\cite{bogunovic-icml17-delete}. The deletion-robust pseudo-Boolean optimization problem with a cardinality constraint can be formally described as
\begin{align}\label{form-deletion}
\max_{x\in\bsp}\min_{z\subseteq x,|z|_1\le d}f(x\setminus z) \quad \text{s.t.} \quad |x|_1\le k,
\end{align}
where $k$ is the cardinality constraint, and $d$ is the maximum number of 1-bits that can be deleted. Note that $f(x)$ is the original objective function to be maximized, while the objective function changes to $\min_{z\subseteq x,|z|_1\le d}f(x\setminus z)$ in the deletion-robust setting. 
Another common type of robust optimization, called worst-case optimization~\cite{krause-jmlr08-rsos}, is to find a solution which is robust against a number of possible objective functions. For example, in spatial monitoring of certain
phenomena, sensors are often used to measure various parameters such as temperature and humidity at the same time, and observations for these parameters are modeled by different functions; the goal is to find a solution of placing sensors which can perform well on all objective functions, i.e., to optimize the worst of all objective functions~\cite{anari-arxiv18-off-online}. In the influence maximization problem, influence of a set of users is measured by a function $\sigma$, which has significant uncertainty due to different models and different parameters; the goal is to optimize a set of functions simultaneously, in which one function is assured to describe the influence process exactly (but which one is not told)~\cite{he-kdd16-robust-inf}. The worst-case pseudo-Boolean optimization problem with a cardinality constraint can be formally described as
\begin{equation}\label{form-worst}
\max_{x\in\bsp}\min_{s\in\{1,2,\dots,m\}}f_s(x) \quad \text{s.t.} \quad |x|_1\le k,
\end{equation}
where $\{f_s\}_{s=1}^m$ are $m$ possible objective functions. Note that the goal to be optimized in the worst-case setting is $\min_{s\in\{1,2,\dots,m\}}f_s(x)$. 

Krause et al.~\cite{krause-jmlr08-rsos} considered the case where $f$ is monotone increasing and satisfies the submodular, i.e., diminishing returns, property, and proposed an algorithm SATURATE which can achieve a solution matching the optimal objective value but with cardinality slightly larger than $k$. SATURATE can apply to both deletion-robust and worst-case scenarios (i.e., Eqs.~\eqref{form-deletion} and \eqref{form-worst}), but the running time for the deletion-robust scenario is exponential in $d$. Thus, Orlin et al.~\cite{orlin-ipco16-deletion} proposed a polynomial-time algorithm achieving an approximation ratio of 0.387 for $d=o(\sqrt{k})$, and Bogunovic et al.~\cite{bogunovic-icml17-delete} further improved to $d=o(k)$ while retaining the approximation guarantee. Bogunovic et al.~\cite{bogunovic-aistats18-nonsub} also considered the deletion-robust problem with non-submodular objective functions.
For the worst-case scenario in Eq.~\eqref{form-worst}, Anari et al.~\cite{anari-arxiv18-off-online} proposed a greedy-style algorithm with an $(1-\epsilon)$-approximation ratio, where $\epsilon \in (0,1)$ relates to the running time of the algorithm as well as the size of the generated solution. Udwani~\cite{udwani-nips18-multi} designed a fast and practical algorithm for the case $m=o(k)$.
He and Kempe~\cite{he-kdd16-robust-inf} applied a modification of SATURATE~\cite{krause-jmlr08-rsos} to the robust influence maximization problem, and showed that an $(1-e)$-approximation ratio can be achieved when enough extra seeds may be selected.

%Powers et al.~\cite{powers-nipsopt16} considered matroid constraint which is an extension of cardinality constraint, and derived a solution that is good only for a fraction of the functions. Their result can also extend to other types of constraints, such as knapsack constraint and multiple matroids.

This paper aims at moving a step towards theoretically analyzing EAs for robust optimization. Particularly, we analyze the expected running time of the (1+1)-EA for robust linear optimization with a cardinality constraint $k$. Both deletion-robust and worst-case settings are considered. That is, the objective functions $f$ in Eqs.~(\ref{form-deletion}) and~(\ref{form-worst}) are linear functions, which have been widely used to examine theoretical properties of EAs~\cite{shi-algo18-dynamic,droste2002analysis,friedrich2018analysis,jansen-analyzing,witt-cpc13,neumann-gecco19-improved}. 
For deletion-robust linear optimization, we also consider two specific instances, i.e., deletion-robust OneMax and deletion-robust BinVal, where the objective functions $f$ in Eq.~(\ref{form-deletion}) are fixed to OneMax and BinVal, respectively. For each concerned robust optimization problem, we derive tight bounds on $d$ or $k$ allowing the (1+1)-EA to find an optimal solution in polynomial running time, which are summarized in Table~\ref{table-runtime}.

\begin{table}[!t]
	\footnotesize
	\centering
	\captionof{table}{For the expected running time (ERT) of the (1+1)-EA solving robust linear optimization problems, the ranges of $d$ or budget $k$ for a polynomial upper bound and a super-polynomial lower bound are shown below. For deletion-robust linear optimization, a polynomial upper bound means that for any $k>d$ (and any linear function), the ERT is polynomial, while a super-polynomial lower bound means that there exists some $k>d$ (and some linear function) such that the ERT is super-polynomial. For worst-case linear optimization, a polynomial upper bound means that for any $m$ linear functions, the ERT is polynomial, while a super-polynomial lower bound means that there exist $m$ linear functions such that the ERT is super-polynomial.
		\protect\\
		Notes: {$0\le c_1=O(1)$, $c_2=\omega(1)$}, $w_1$ denotes the maximum weight of the linear function, $\delta$ denotes the minimum difference of two different weights of the linear function and $\delta:=1$ if all the weights are the same, and $w_{\max}$ denotes the maximum weight of all $m$ linear functions. 
	}\label{table-runtime}
	\begin{tabular}{p{2.35cm}|p{2.65cm}p{2.25cm}p{0.9cm}|p{2.1cm}p{1.2cm}p{1cm}}
		\toprule
		%		& & \\[-8pt]
		Problem & \multicolumn{3}{l|}{Polynomial upper bound} & \multicolumn{3}{l}{Super-polynomial lower bound} \\
		\midrule
		%		& &  \\[-7pt]
		\multirow{2}{*}{\makecell[l]{Deletion-robust\\OneMax}} & $d=o(n)$
		&  $O(n\log n)$ & \multirow{2}{*}{{\footnotesize{[Thm.\ref{del-om-upper}]}}}
		& &  \\[2pt]
		%	&& \\[-8pt]
		& {$d\le  n/2+c_1\sqrt{n\log n}$} & {$O(n^{7c_1^2+2})$} & & \multirow{2}{*}{{\makecell[l]{$d=$\\$n/2+c_2\sqrt{n\log n}$}}} &   \multirow{2}{*}{${n^{2c_2^2}/4}$} & \multirow{2}{*}{{\footnotesize{[Thm.\ref{del-om-lower},\ref{thm-del-bv-lower}]}}}\\[2pt]
		\cmidrule{1-4}
		\makecell[l]{Deletion-robust\\BinVal} & {$d\le  n/2+c_1\sqrt{n\log n}$} & {$O(n^{7c_1^2+2}$$+n^2\log n)$} & {\footnotesize{[Thm.\ref{thm-del-bv-upper}]}} & &\\[2pt]
		%	&&\\[-8pt]
		\midrule
		\multirow{2}{*}{\makecell[l]{Deletion-robust\\linear optimization}} & \multirow{2}{*}{$d= O(1)$} & polynomial in $n,$& \multirow{2}{*}{{\footnotesize{[Thm.\ref{thm-del-linear-upper}]}}}&  \multirow{2}{*}{$d= \omega(1)$} &  \multirow{2}{*}{\makecell[l]{super-\\polynomial}} & \multirow{2}{*}{{\footnotesize{[Thm.\ref{thm-del-linear-lower}]}}}\\
		& & $\log w_1$ and $1/\delta$ & \\
		%		& &  \\[-8pt]
		\midrule
		\multirow{3}{*}{\makecell[l]{Worst-case linear\\ optimization}} & $k= O(1)$ & polynomial & \multirow{3}{*}{{\footnotesize{[Thm.\ref{thm-wc-upper}]}}}& \multirow{3}{*}{\makecell[l]{$k=$\\ $\omega(1)\cap n-\omega(1)$}} & \multirow{3}{*}{\makecell[l]{super-\\polynomial}} & \multirow{3}{*}{{\footnotesize{[Thm.\ref{thm-wc-lower}]}}}\\[4pt]
		& \multirow{2}{*}{$k= n-O(1)$} & polynomial in & &\\
		& & $n$ and $w_{\max}$ & & &	\\\bottomrule
		\multicolumn{5}{l}{}\\[-7pt]
	\end{tabular}
\end{table}

From the results, we can find that the (1+1)-EA can efficiently solve the deletion-robust OneMax and deletion-robust BinVal problems when $d$ is not very large. For example, when $d=o(n)$, the expected running time for deletion-robust OneMax is $O(n\log n)$, which is the same as the known bound of the (1+1)-EA solving the OneMax problem in exact environments~\cite{droste2002analysis}. Even for the general deletion-robust and worst-case linear optimization problems, the (1+1)-EA can be efficient when $d=O(1)$ and $k=O(1)$, respectively. As the practical value of $k$ is often not too large (note that $d$ is also not too large because $d<k$), the results disclose the potential of EAs for robust optimization. We also note that the performance of the (1+1)-EA degrades as deletion-robust linear optimization is extended from specific cases, i.e., OneMax and BinVal, to general cases, suggesting that more complicated EAs may be desired to deal with real-world robust optimization.

The rest of this paper is organized as follows. Section~\ref{sec-preliminary} introduces some preliminaries. The running time analysis for deletion-robust and worst-case linear optimization is presented in Sections~\ref{sec-deletion} and \ref{sec-worst}, respectively. Section~\ref{sec-conclusion} concludes the paper.

\section{Preliminaries}\label{sec-preliminary}

In this section, we first introduce the considered problem and algorithm, i.e., robust linear optimization and the (1+1)-EA, respectively, and then present the analysis tools that we use throughout this paper.

\subsection{Robust Linear Optimization}

As discussed before, deletion-robust and worst-case linear optimization under a cardinality constraint, i.e., Eqs.~(\ref{form-deletion}) and~(\ref{form-worst}) with linear objective functions, are considered in this paper. The linear problem with a cardinality constraint, as presented in Definition~\ref{def-linear}, aims to maximize the weighted sum of a bit string with the constraint that the number of 1-bits is no larger than $k$. We assume the weights are all no smaller than 1 and $w_1\ge w_2\ge ...\ge w_n$. It is clear that $1^{k}0^{n-k}$, i.e., the string with $k$ leading 1-bits and $n-k$ trailing 0-bits, is an optimal solution. Note that this problem has been studied in \cite{friedrich2018analysis} and its dynamic version has been studied in \cite{shi-algo18-dynamic}.% expected running time of the (1+1)-EA is  \cite{friedrich2018analysis}.

\begin{definition}[Linear Problem with A Cardinality Constraint]\label{def-linear}
	Given $n$ weights $\{w_i\}_{i=1}^n$ where $w_1\ge w_2\ge ...\ge w_n \geq 1$, and a budget $k\le  n$, to find a binary solution $x\in \{0,1\}^n$ such that
	\begin{equation*}
	\max_{x\in\bsp} \sum_{i=1}^nw_ix_i \quad \text{s.t.} \quad |x|_1\le k, 
	\end{equation*} 
	where $x_i$ denotes the $i$-th bit of $x\in \bsp$.
\end{definition}

The deletion-robust linear optimization problem is presented in Definition~\ref{def-deletion-linear}. When the context is clear, let $F(x)=\min_{z\subseteq x,|z|_1\le d}$ $\sum_{i=1}^{n}w_i(x_i-z_i)$, which is actually the weighted sum of the solution generated by deleting the leftmost $d$ 1-bits of $x$, as $w_i$ decreases with $i$. Then the problem can be viewed as maximizing $F(\cdot)$. Note that $1^{k}0^{n-k}$ is still an optimal solution. 

\begin{definition}[Deletion-robust Linear Optimization]\label{def-deletion-linear}
	Given $n$ weights $\{w_i\}_{i=1}^n$ where $w_1\ge w_2\ge ...\ge w_n \geq 1$, a budget $k\le  n$ and a parameter $d<k$, to find a binary solution $x\in \{0,1\}^n$ such that
	\begin{equation} \label{eq-deletion-linear}
	\max_{x\in\bsp}\min_{z\subseteq x,|z|_1\le d}\sum_{i=1}^{n}w_i(x_i-z_i) \quad \text{s.t.} \quad |x|_1\le k. \end{equation}
\end{definition}

We also consider two specific instances of deletion-robust linear optimization, deletion-robust OneMax and deletion-robust BinVal presented in Definitions~\ref{def-deletion-onemax} and~\ref{def-deletoin-binval}, respectively. That is, the objective functions are specified to be OneMax and BinVal, respectively, which are two extreme instances of linear functions. For OneMax, all weights have the same value 1; for BinVal, $\forall i: w_i=2^{n-i}$ is larger than the sum of $\{w_j\}_{j=i+1}^n$. It is clear that any solution with $k$ 1-bits is optimal for deletion-robust OneMax, and $1^{k}0^{n-k}$ is the unique optimal solution for deletion-robust BinVal.

\begin{definition}[Deletion-robust OneMax]\label{def-deletion-onemax}
	Given a budget $k\le  n$ and a parameter $d<k$, to find a binary solution $x\in \{0,1\}^n$ such that
	\begin{equation}\label{eq-deletion-onemax}
	\max_{x\in\bsp}\min_{z\subseteq x,|z|_1\le d}\sum_{i=1}^{n}(x_i-z_i) \quad \text{s.t.}\quad |x|_1\le k.
	\end{equation}
\end{definition}

\begin{definition}[Deletion-robust BinVal]\label{def-deletoin-binval}
	Given a budget $k\le  n$ and a parameter $d<k$, to find a binary solution $x\in \{0,1\}^n$ such that
	\begin{equation}\label{eq-deletion-binval}
	\max_{x\in\bsp}\min_{z\subseteq x,|z|_1\le d}\sum_{i=1}^{n}2^{n-i}(x_i-z_i) \quad \text{s.t.} \quad |x|_1\le k. 
	\end{equation} 
\end{definition}

The worst-case linear optimization problem is presented in Definition~\ref{def-worst-linear}. When the context is clear, let $F(x)= \min_{s\in\{1,2,\dots,m\}}f_s(x)$, which is the {minimum at $x$ of $m$} given linear functions. Then the problem can be viewed as maximizing $F(\cdot)$. For each linear function, the weights are only assumed to be no smaller than 1, and can be non-decreasing with $i$. Note that $w_{s,i}$ denotes the $i$-th weight of the $s$-th linear function.
%Note that the objective function is the goal is to

\begin{definition}[Worst-case Linear Optimization]\label{def-worst-linear}
	Given $m$ linear functions $\{f_s\}_{s=1}^m$ where $f_s(x)=\sum_{i=1}^n{w_{s,i}}x_i, {w_{s,i}} \geq 1$, and a budget $k\le  n$, to find a binary solution $x\in \{0,1\}^n$ such that
	\begin{equation}\label{eq-worst-linear}
	\max_{x\in\bsp}\min_{s\in\{1,2,\dots,m\}}f_s(x) \quad \text{s.t.}\quad |x|_1\le k.
	\end{equation} 
\end{definition}
%For solving robust linear optimization problems in Definitions~\ref{def-deletion-linear}-\ref{def-worst-linear}, we assume that $F(x)$ can be obtained. Note that we  %For solving robust linear optimization problems, we assume that $F(x)$ can be obtained.
%Under the cardinality constraint $k$, the fitness of a solution $x$, denoted as $g(x)$, 

\subsection{(1+1)-EA}
In this paper, we consider the simple (1+1)-EA and assume that $F(x)$ can be obtained when evaluating $x$. Note that such setting was also employed in~\cite{he-kdd16-robust-inf}, where the greedy algorithm maximizing the minimum of given functions achieves good empirical performance.

The (1+1)-EA as described in Algorithm~\ref{(1+1)-EA} maintains only one solution, and repeatedly improves the current solution by using bit-wise mutation in line~3 and selection in line~4. Note that $g(x)$ is the fitness of a solution $x$ defined~as% Eq.~\eqref{def-fitness}. % iteratively tries to generate one better solution by bit-wise mutation. 
%is calculated as follows: 
\begin{equation}\label{def-fitness}
g(x)=
\begin{cases}
k-|x|_1 & \text{if } |x|_1>k;\\
F(x) & \text{otherwise}.
\end{cases}
\end{equation}
That is, if $x$ violates the constraint, $g(x)$ is the degree of constraint violation; otherwise, $g(x)=F(x)$. Note that for the problems examined in this paper, the fitness of a feasible solution is always larger than that of an infeasible one. That~is, the common strategy ``superiority of feasible points"~\cite{deb2000efficient} for constrained optimization is employed. 
The running time of the (1+1)-EA is defined as the number of fitness evaluations needed to find an optimal solution for the first time.% with respect to $F(x)$ 
\begin{algorithm}
	\caption{(1+1)-EA}\label{(1+1)-EA} 
		\begin{algorithmic}
		\STATE 	Given a fitness function $g: \{0,1\}^n \rightarrow \mathbb{R}$ to be maximized, the procedure:
		\STATE  1. Let $x$ be a uniformly randomly chosen solution.
		\STATE 2. Repeat until the termination condition is met
		\STATE 3. \quad $x'\!:=\!$ flip each bit of $x$ independently with prob. $1/n$. 
		\STATE  4. \quad if {$g(x') \geq g(x)$} \quad then $x:=x'$.
	\end{algorithmic}
\end{algorithm}

\subsection{Analysis Tools}
The process of the (1+1)-EA solving any pseudo-Boolean function can be directly modeled as a Markov chain $\{\xi_t\}^{+\infty}_{t=0}$. We only need to take the solution space $\{0,1\}^n$ as the chain's state space, i.e., $\xi_t \in \mathcal{X}=\{0,1\}^n$, and take all optimal solutions as the chain's target state space $\mathcal{X}^*$. Given a Markov chain $\{\xi_t\}^{+\infty}_{t=0}$ and $\xi_{\hat{t}}=x$, we define its \emph{first hitting time} (FHT) as $\tau=\min\{t \mid \xi_{\hat{t}+t} \in \mathcal{X}^*,t\geq0\}$. The mathematical expectation of $\tau$, $\mathrm{E}(\tau \mid \xi_{\hat{t}}=x)=\sum\nolimits^{+\infty}_{i=0} i\cdot\mathrm{P}(\tau=i \mid \xi_{\hat{t}}=x)$, is called the \emph{expected first hitting time} (EFHT) starting from $\xi_{\hat{t}}=x$. If $\xi_{0}$ is drawn from a distribution $\pi_{0}$, $\mathrm{E}(\tau \mid \xi_{0}\sim \pi_0) = \sum\nolimits_{x\in \mathcal{X}} \pi_{0}(x)\mathrm{E}(\tau \mid \xi_{0}=x)$ is called the EFHT of the Markov chain over the initial distribution $\pi_0$. Thus, the expected running time of the (1+1)-EA starting from $\xi_0 \sim \pi_0$ is equal to $1+\mathrm{E}(\tau \mid \xi_{0} \sim \pi_0)$. Note that we consider the expected running time of the (1+1)-EA starting from a uniform initial distribution in this paper.

We will use the additive drift theorem (i.e., Theorem~\ref{additive-drift}) as well as the multiplicative drift theorem (i.e., Theorem~\ref{multiplicative-drift}) to analyze the EFHT of Markov chains. To use them, a function $V(x)$ has to be constructed to measure the distance of a state $x$ to the target state space $\mathcal{X}^*$. The distance function $V(x)$ satisfies that $V(x\in \mathcal{X}^*)=0$ and $V(x\notin \mathcal{X}^*)>0$. Then, we need to examine the progress on the distance to $\mathcal{X}^*$ in each step, i.e., $\expect{V(\xi_t)-V(\xi_{t+1}) \mid \xi_t}$. For additive drift analysis, an upper bound on the EFHT can be derived through dividing the initial distance by a lower bound on the progress. Multiplicative drift analysis is much easier to use when the progress is roughly proportional to the current distance to the optimum.

\begin{theorem}[Additive Drift~\cite{he2001drift}]\label{additive-drift}
	Given a Markov chain $\{\xi_t\}^{+\infty}_{t=0}$ and a distance function $V(x)$, if for any $t \geq 0$ and any $\xi_t$ with $V(\xi_t) > 0$, there exists $c>0$ such that \begin{equation*}
	\mathrm{E}(V(\xi_t)-V(\xi_{t+1}) \mid \xi_t) \geq c,
	\end{equation*} then the EFHT satisfies $\mathrm{E}(\tau \mid \xi_0) \leq V(\xi_0)/c.$
\end{theorem}

\begin{theorem}[Multiplicative Drift~\cite{doerr:etal:GECCO10}]\label{multiplicative-drift}
	Given a Markov chain $\{\xi_t\}^{+\infty}_{t=0}$ and a distance function $V(x)$, if for any $t \geq 0$ and any $\xi_t$ with $V(\xi_t) > 0$, there exists $c>0$ such that \begin{equation*}
	\expect{V(\xi_t)-V(\xi_{t+1}) \mid \xi_t} \geq c \cdot V(\xi_t),
	\end{equation*} then the EFHT satisfies $\mathrm{E}(\tau \mid \xi_0) \leq (1+\log (V(\xi_0)/V_{\min}))/c$, where $V_{\min}$ denotes the minimum among all possible positive values of \,$V$.
\end{theorem}

\section{Deletion-robust Linear Optimization}\label{sec-deletion}

In this section, we analyze the expected running time of the (1+1)-EA for deletion-robust OneMax (i.e., Eq.~\eqref{eq-deletion-onemax}), deletion-robust BinVal (i.e., Eq.~\eqref{eq-deletion-binval}) and general deletion-robust linear optimization (i.e., Eq.~\eqref{eq-deletion-linear}), respectively. The ranges of $d$ for a polynomial upper bound and a super-polynomial lower bound are all derived.

\subsection{Deletion-robust OneMax}

For the (1+1)-EA solving deletion-robust OneMax (i.e., Eq.~\eqref{eq-deletion-onemax}), Theorems~\ref{del-om-upper} and~\ref{del-om-lower} show that the tight range of $d$ allowing polynomial running time is $[1,n/2+O(\sqrt{n\log n})]$. The reason for the effectiveness of the (1+1)-EA when $d$ is not too large is as follows. For any solution $x$, if $|x|_1> k$, the optimization procedure is similar to that of the (1+1)-EA minimizing the OneMax function, and the number of 1-bits of the solution quickly decreases to at most $k$; if $|x|_1\le d$, the fitness is 0 and any offspring solution will be {accepted}, thus the (1+1)-EA performs like a random walk over $\bsp$, and the number of 1-bits of the solution can increase to $d+1$ in polynomial running time for $d=n/2+O(\sqrt{n\log n})$ or $d\le n/2$. Thus, the (1+1)-EA can efficiently find a solution $x$ with $d+1\le |x|_1\le k$ and then quickly find an optimal solution. 

To examine the behavior of the (1+1)-EA performing a random walk, we present Lemma~\ref{lem-randomwalk}, which gives the expected running time for the (1+1)-EA on a plateau to reach a solution with a sufficient number of 1-bits. The proof of Lemma~\ref{lem-randomwalk} is accomplished by applying Theorem~\ref{additive-drift}, i.e., the additive drift theorem.

\begin{lemma}\label{lem-randomwalk}
	For any $d= n/2+r$, $1\le r=O(\sqrt{n\log n})$, the (1+1)-EA which always accepts the offspring solution in each iteration, can find a solution with more than $d$ 1-bits in expected running time $O(n^2\cdot e^{7r^2/n})$, i.e., polynomial. For $d=o(n)$, the expected running time is $o(n)$.
\end{lemma}
\begin{proof}
	We use Theorem~\ref{additive-drift} for the proof. Note that $\mathcal{X}=\{0,1\}^n$ and $\mathcal{X}^*=\{x \in \{0,1\}^n \mid |x|_1>d\}$. First, we consider $d=n/2+r$. The distance function is constructed as
	%For any solution $x$ with $|x|_1 \ge n/2$, the expected change in the number of 1-bits is negative. To derive a positive drift, we design a specific distance function to enlarge the effect of the increased 1-bits. 
	\begin{equation*}
	V(x)=\left\{
	\begin{aligned}
	&{(1+7r)\blue{\Big(}1+\frac{7r}{n}\blue{\Big)^{r}}+\blueBl\frac{n}{2}-|x|_1\blueBr\frac{7r}{n+7r}}  &\phantom{=}& |x|_1  <\frac{n}{2}, \\
	& {(1+7r)\blueBl 1+\frac{7r}{n}\blue{\Big)^r}-\blueBl 1+\frac{7r}{n}\blue{\Big)^{|x|_1-n/2}}+1} &\phantom{=}& \frac{n}{2}  \le |x|_1\le d,\\
	&0 &\phantom{=}& |x|_1  >d.
	\end{aligned}\right.
	\end{equation*}
	The design of the distance function is to make a positive drift for the concerned process. When $|x|_1<n/2$, the expected change in the number of one bits is positive, and thus it is sufficient that the distance decreases linearly with $|x|_1$. When $n/2\le |x|_1\le d$, the expected change in the number of one bits is negative, and thus it is necessary to make the distance decrease exponentially with $|x|_1$, such that the gain of flipping a single 0-bit is large enough. When $|x|_1>d$, the target state space is reached, and thus the distance is 0. The coefficients in the distance function are constructed carefully to make a good tradeoff between the positive drift and the initial distance, and thus to make the derived upper bound on the running time as tight as possible.\\
	As $V(x)$ depends only on the number of 1-bits of a solution, we denote $V_i$ as the distance of any solution with $i$ 1-bits, i.e., $\forall x$ with $|x|_1=i$, $V(x)=V_i$. Then, we have $V_{i-1}-V_i =(7r)/(n+7r)$ for $1\le i\le n/2$, $V_{i-1}-V_i=(7r/(n+7r))\cdot(1+7r/n)^{i-n/2}$ for  $n/2<i\le d$, and $V_d-V_{d+1}=7r(1+7r/n)^{r}+1$. Thus, $V_i$ decreases with $i$ for $0\le i\le n$, $V_{i-1}-V_i$ increases with $i$ for $1\le i\le d+1$, and $V_i=0$ if and only if $i>d$.
	%For $x$ with $|x|_1<n/2$, $V(x)>(1+7r)(1+7r/n)^r$; for $x$ with $n/2\le |x|_1\le d$, $V(x)\le (1+7r)(1+7r/n)^r$. for It can be verified  that $V(x \in \mathcal{X}^*)=0$ and $V(x \notin \mathcal{X}^*)>0$. 
	%(1+\frac{20c}{n})^{c}20c/n & j=d+1.
	% Note that $\forall j \geq d+1: V_i=0$. 
	%It can be verified that $(V_i-V_{i+1})/(V_{i-1}-V_i)=1$ for $1\le j< n/2$, and $(V_i-V_{i+1})/(V_{i-1}-V_i)=1+7r/n$ for $n/2\le i\le d$.
	%Meanwhile, $V_d-V_{d+1}=(1+\frac{7r}{n})^{r}7r+1\ge n\cdot (V_{i-1}-V_i)$, for $1\le i\le d+1$. 
	
	Next we examine $\mathrm{E}(V(\xi_t)-V(\xi_{t+1}) \mid \xi_t=x)$. Assume that currently $|x|_1=i\le d$. 
	Let $\mathrm{P_{mut}}(j)$ denote the probability that a solution with $j$ 1-bits is generated from $x$ by mutation. Thus, we have
	\begin{align*}
	&\mathrm{E}(V(\xi_t)-V(\xi_{t+1}) \mid \xi_t=x)\\
	&=\sum\nolimits_{j=0}^{n}\mathrm{P_{mut}}(j)\cdot (V_i-\!V_j)\\
	&=\sum\nolimits_{j=0}^{i-1}\mathrm{P_{mut}}(j)\cdot(V_i-V_j)+\mathrm{P_{mut}}(i)\cdot (V_i-V_i)+\sum\nolimits_{j=i+1}^{d}\mathrm{P_{mut}}(j)\cdot (V_i-V_j)+\sum\nolimits_{j=d+1}^{n}\mathrm{P_{mut}}(j)\cdot (V_i-V_j)\\
	&\ge\sum\nolimits_{j=0}^{i-1}\mathrm{P_{mut}}(j)\cdot(V_i-V_{i-1}+V_{i-1}-\ldots -V_j)+\sum\nolimits_{j=i+1}^{d}\mathrm{P_{mut}}(j)\cdot (V_i-V_{i+1}+V_{i+1}-\ldots -V_j)\\
	&\quad+\sum\nolimits_{j=d+1}^{n}\mathrm{P_{mut}}(j)\cdot n\cdot \blueBl 1+\frac{7r}{n}\blueBr\cdot(V_{d-1}-V_{d})\\
	&\ge \sum\nolimits_{j=0}^{i-1}\mathrm{P_{mut}}(j)\cdot (V_{i}-V_{i-1})\cdot(i-j)+\sum\nolimits_{j=i+1}^{d}\mathrm{P_{mut}}(j)\cdot (V_{i}-V_{i+1})\cdot(j-i)\\
	&\quad+\sum\nolimits_{j=d+1}^{n}\mathrm{P_{mut}}(j)\cdot \blueBl 1+\frac{7r}{n}\blueBr\cdot(V_{d-1}-V_{d})\cdot(j-i),
	\end{align*}
	where the first inequality holds by $V_i-V_{j}\ge V_d-V_{d+1}= (n+7r)(V_{d-1}-V_d)+1$ for $i\le d$ and $j\ge d+1$, and the  second inequality holds by the monotonicity of $V_{i-1}-V_i$ for $1\le i\le d$. \\
	If $i< n/2$,  we have
	\begin{align*}\label{eq-rw}
	\mathrm{E}&(V(\xi_t)-V(\xi_{t+1}) \mid \xi_t=x)\ge \frac{7r}{n+7r}\cdot \sum\nolimits_{j=0}^{n}\mathrm{P_{mut}}(j)\cdot (j-i) =\frac{7r}{n+7r}\cdot \frac{n-2i}{n}
	=\Omega\blueBl\frac{r}{n^2}\blueBr ,
	\end{align*}
	where the equality holds because $\sum\nolimits_{j=0}^{n}\mathrm{P_{mut}}(j)\cdot (j-i)$ denotes the expected number of increased 1-bits after mutation, which is the expected number of flipped 0-bits (i.e., $(n-i)/n$) minus the expected number of flipped 1-bits (i.e., $i/n$).
	% $V_i-V_{i+1}\ge V_{i-1}-V_i$ for $1\le i\le d$.\\
	If $i\ge n/2$, we have
	\begin{align*}
	&\mathrm{E}(V(\xi_t)-V(\xi_{t+1})\mid \xi_t=x)\\
	&\ge 
	\sum\nolimits_{j=0}^{i-1}\mathrm{P_{mut}}(j)\cdot (V_{i-1}-V_{i})\cdot(j-i)+\sum\nolimits_{j=i+1}^{n}\mathrm{P_{mut}}(j)\cdot \blueBl 1+\frac{7r}{n}\blueBr\cdot(V_{i-1}-V_{i})\cdot(j-i)\\
	&\ge  (V_{i-1}-V_i)\cdot \left(\sum\nolimits_{j=0}^{n}\mathrm{P_{mut}}(j)\cdot (j-i)+\frac{7r}{n}\cdot \sum\nolimits_{j=i+1}^{n}\mathrm{P_{mut}}(j)\cdot (j-i)\right).
	\end{align*}
	
	Now we examine the lower bound for $\sum\nolimits_{j=i+1}^{n}\mathrm{P_{mut}}(j)\cdot (j-i)$.
	\begin{align*}
	&\sum\nolimits_{j=i+1}^{n}\mathrm{P_{mut}}(j)\cdot (j-i)\\
	&\ge \sum\nolimits_{j=i+1}^{i+4}\mathrm{P_{mut}}(j)\cdot (j-i)
	= \sum\nolimits_{j=1}^{4}{\left(1-\frac{1}{n}\right)^{n-j}\cdot \left(\frac{1}{n}\right)^{j}}\cdot \binom{n-i}{j}\cdot j\\
	&\ge\sum\nolimits_{j=1}^{4}\frac{j}{en^{j}}\cdot\frac{(n-i-3)^{j}}{j!}\ge \sum\nolimits_{j=1}^{4}\frac{(1/2-o(1))^{j}}{e(j-1)!}\\
	&=\frac{1}{e}\cdot\Big(\frac{1}{2}+\frac{1}{4}+\frac{1}{16}+\frac{1}{96}-o(1)\Big)=\frac{79}{96e}-o(1),
	\end{align*}
	where the third inequality is by $i\le d=n/2+r=n/2+O(\sqrt{n\log n})$. Thus, 
	\begin{align*}
	&\mathrm{E}(V(\xi_t)-V(\xi_{t+1})\mid \xi_t=x)\\
	&\ge (V_{i-1}-V_i)\cdot \left(\frac{n-2i}{n}+\frac{7r}{n}\cdot \blueBl\frac{79}{96e}-o(1)\blueBr\right) \\
	&\ge (V_{i-1}-V_i)\cdot \frac{r}{n}\cdot \left(7\cdot \frac{79}{96e}-o(1)-2\right)\ge \frac{7r}{n+7r}\cdot\Omega\blueBl\frac{1}{n}\blueBr=\Omega\blueBl\frac{r}{n^2}\blueBr,
	\end{align*}
	where the second inequality is by $i\le d= n/2+r$.\\
	Combining the above two cases, we have
	\begin{equation*}
	\mathrm{E}(V(\xi_t)-V(\xi_{t+1}) \mid \xi_t=x)\ge \Omega\blueBl\frac{r}{n^2}\blueBr.
	\end{equation*}
	
	Note that 
	\begin{align*}
	V(\xi_0)\le  (1+7r)\left(1+\frac{7r}{n}\right)^{r}+\frac{n}{2}\cdot \frac{7r}{n+7r} \le (1+7r)e^{7r^2/n}+\frac{7r}{2}=O\left(r\cdot e^{7r^2/n}\right),
	\end{align*}
	where the second inequality holds by $\forall a\in \mathbb{R}$: $1+a\le e^a$. %$c=O(\sqrt{n\log n})$ and
	%$$\left(1+\frac{20c}{n}\right)^c= \left(1+\frac{20c^2/n}{c}\right)^c\le .$$
	Thus, by Theorem~\ref{additive-drift}, the expected running time is $\mathrm{E}(\tau \mid \xi_0) \leq  V(\xi_0)/\Omega(r/n^2)=O(n^2\cdot e^{7r^2/n})$, i.e., polynomial.
	
	Then, we consider $d=o(n)$. The distance function is constructed as 
	\begin{equation*}
	V(x)=\begin{cases}
	d+1-|x|_1 & |x|_1\le d, \\
	0 & |x|_1>d.
	\end{cases}
	\end{equation*}
	Next we examine $\mathrm{E}(V(\xi_t)-V(\xi_{t+1}) \mid \xi_t=x)$. To derive an upper bound on the expected increase of $V(\cdot)$, we pessimistically assume that the 0-bits of $x$ are not flipped. Note that the expected number of flipped  1-bits is at most $|x|_1/n$, thus the expected increase of $V(\cdot)$ is at most $|x|_1/n$. To derive a lower bound on the expected decrease of $V(\cdot)$, we only need to consider that one 0-bit of $x$ is flipped, whose probability is $(n-|x|_1)/n\cdot (1-1/n)^{n-1}\ge (n-|x|_1)/(en)$. Then we have
	\begin{equation*}
	\mathrm{E}(V(\xi_t)-V(\xi_{t+1}) \mid \xi_t=x)\ge (n-|x|_1)/(en)-|x|_1/n=\Omega(1).
	\end{equation*}
	Note that $V(\xi_0)\le d+1=o(n)$. Thus, by Theorem~\ref{additive-drift}, the expected running time is $\mathrm{E}(\tau \mid \xi_0) \leq V(\xi_0)\cdot O(1)= o(n)$.
 \end{proof}

Next, we prove Theorem~\ref{del-om-upper} by applying Lemma~\ref{lem-randomwalk} and Theorem~\ref{multiplicative-drift}, i.e., the multiplicative drift theorem.

\begin{theorem}\label{del-om-upper}
	If {$d\le  n/2+c\sqrt{n\log n}$, $0\le c=O(1)$}, then $\forall k>d$, the expected running time of the (1+1)-EA for deletion-robust OneMax is {$O(n^{7c^2+2})$}, i.e., polynomial. Furthermore, if $d=o(n)$, the expected running time is $O(n\log n)$.%If $d<n/2$, the expected running time is $O(n^2)$.
\end{theorem}
\begin{proof}
	We divide the optimization procedure into two phases: (1) starts after initialization and finishes upon finding a solution $x$ with $d+1 \leq |x|_1 \leq k$; (2) starts after (1) and finishes upon finding an optimal solution.
	
	For phase (1), we further consider two subphases.\\
	(1a) A solution $x$ with $|x|_1\le k$ is found. We pessimistically assume that the initial solution has more than $k$ 1-bits, i.e., $|\xi_0|_1>k$. Then, the concerned procedure is the same as that of the (1+1)-EA minimizing an unconstrained OneMax function according to the definition of the fitness function (i.e., Eq.~\eqref{def-fitness}); thus, the expected running time until the number of 1-bits decreases to at most $k$ is $O(n\log n)$~\cite{doerr:etal:GECCO10}. Note that after phase~(1a), the (1+1)-EA will always maintain a solution with at most $k$ 1-bits.\\
	(1b) Starts after phase~(1a), and finishes upon finding a solution $x$ with $d+1 \leq |x|_1 \leq k$. We pessimistically assume that after phase~(1a), the solution has at most $d$ 1-bits. % i.e., $|\xi_0|_1\le d$. 
	We first analyze the case of {$n/2+1\le d\le n/2+c\sqrt{n\log n}$. Let $r=d-n/2$}. We denote a \textit{good jump} and a \textit{successful jump} as:
	\begin{itemize}
		\item a solution $x'$ with $|x|_1\ge d+1$ is generated from a solution $x$ with $|x|_1\le d$ by bit-wise mutation,
		\item a solution $x'$ with $d+1\le |x|_1\le k$ is generated from a solution $x$ with $|x|_1\le d$ by bit-wise mutation,
	\end{itemize}
	respectively. %Note that a good jump is not necessarily a successful jump, because $x'$ may have more than $k$ bits and thus, it will not be accepted. 
	Because the fitness is the same for solutions with the number of 1-bits no larger than $d$, the expected running time until the number of 1-bits increases to at least $d+1$ (i.e., a good jump happens) is {$O(n^2\cdot e^{7r^2/n})=O(n^{7c^2+2})$} by Lemma~\ref{lem-randomwalk}. For any $x$ with $|x|_1\le d$, the probability of a good jump is at most $\binom{n-|x|_1}{d+1-|x|_1}\cdot (1/n)^{d+1-|x|_1}$,  because $d+1-|x|_1$ 0-bits need to be flipped; the probability of a successful jump is at least $\binom{n-|x|_1}{d+1-|x|_1}\cdot (1/n)^{d+1-|x|_1}(1-1/n)^{n-(d+1-|x|_1)}\ge \binom{n-|x|_1}{d+1-|x|_1}\cdot (1/n)^{d+1-|x|_1}\cdot 1/e$, because it is sufficient to flip exact $d+1-|x|_1$ 0-bits. Thus, the probability of a good jump being successful is at least $1/e$, which implies that the expected number of good jumps needed to produce a successful jump is at most $e$. Therefore, the expected running time until a successful jump happens is {$O(n^{7c^2+2})$}, which is actually the expected running time of phase~(1b). 
	For the case of {$d< n/2+1$}, the expected running time is not greater than that for $d'=n/2+1$, because finding a solution $x$ with $d'+1\le |x|_1\le k$ implies $d+1\le |x|_1\le k$. Thus, the expected running time of phase~(1b) is {$O(n^2\cdot e^{7/n})=O(n^2)$, which implies an upper bound of $O(n^{7c^2+2})$}. \\
	Combining phases~(1a) and (1b), we can derive that the expected running time of phase~(1) is ${O(n^{7c^2+2})}$.
	
	Consider phase (2). After phase~(1), the (1+1)-EA will always maintain a solution $x$ with $d+1 \leq |x|_1 \leq k$. We use Theorem~\ref{multiplicative-drift} to analyze the expected running time until finding an optimal solution, which has $k$ 1-bits. The distance function is constructed as
	\begin{equation*}
	V(x)=k-|x|_1 .
	\end{equation*}
	We examine $\mathrm{E}(V(\xi_t)-V(\xi_{t+1}) \mid \xi_t=x)$. Note that $V(\cdot)$ will not increase, because $|x|_1$ never decreases. To decrease $V(\cdot)$, i.e., to increase $|x|_1$, it is sufficient that exactly one 0-bit of $x$ is flipped, whose probability is $(n-|x|_1)/n\cdot (1-1/n)^{n-1}\ge (k-|x|_1)/(en)$. Thus, we have 
	\begin{equation*}
	\mathrm{E}(V(\xi_t)-V(\xi_{t+1}) \mid \xi_t=x)\ge \frac{k-|x|_1}{en}=\frac{V(x)}{en}.
	\end{equation*}
	By Theorem~\ref{multiplicative-drift}, the expected running time is at most $en(1+\log (k-d))=O(n\log k)$.
	
	Combining phases~(1) and~(2), the expected running time of the (1+1)-EA solving deletion-robust OneMax is {$O(n^{7c^2+2}) + O(n\log k)=O(n^{7c^2+2})$}, i.e., polynomial.
	
	For the other case (i.e., $d=o(n)$), the main difference is phase~(1b), whose expected running time is $o(n)$ instead of ${O(n^{7c^2+2})}$. Then we can derive that the expected running time of the (1+1)-EA for deletion-robust OneMax is $O(n\log n)+o(n)+O(n\log k)=O(n\log n)$.
	% 	Combining the two cases, the expected running time is $O(cn^2\cdot e^{20c^2/n}+n^2)$, i.e., polynomial.
\end{proof}

Now, we present Theorem~\ref{del-om-lower}, which shows that the expected running time is super-polynomial when $d=n/2+\omega(\sqrt{n\log n})$. %From the following analysis, we can find the reason why the (1+1)-EA is ineffective for large $d$. 
The proof intuition is as follows. To find the optimum, the (1+1)-EA needs to perform a random walk on a plateau, consisting of all solutions with the number of 1-bits no larger than $d$. When $d$ is large, the size of the plateau can be quite large, and the expected increase of 1-bits can be negative; thus, the (1+1)-EA is inefficient.
\begin{theorem}\label{del-om-lower}
	If {$d= n/2+c\sqrt{n\log n}$, $c=\omega(1)$}, then $\forall k>d$, the expected running time of the (1+1)-EA for deletion-robust OneMax is at least ${n^{2c^2}/4}$, i.e., super-polynomial.
\end{theorem}
\begin{proof}
	We consider the expected running time of finding a solution with $|x|_1>d$.
	Let a Markov chain $\{\xi_t\}^{+\infty}_{t=0}$ model the concerned evolutionary process. That is, $\xi_t$ corresponds to the solution after running $t$ iterations of the (1+1)-EA. Note that for any $x$ with $|x|_1\le d$, $F(x)=0$. Thus, the optimization procedure is analogous to a random walk, that is, the offspring solution will always be accepted. Note that for $t=0$, each bit takes 1 or 0 with equal probability (i.e., 1/2), and any flipping of the bit will be accepted in the following iteration, thus for any $t > 0$, the distribution of each bit is the same as that of $t=0$. 
	%the distribution of $\xi_t$ is a uniform distribution over $\{0,1\}^n$, and 
	Thus, $\pr{|\xi_t|_1>d}=\pr{|\xi_0|_1>d}$. By Hoeffding's inequality, we have
	\begin{align*}
	\pr{|\xi_0|_1> d}=\mathrm{P}\left(|\xi_0|_1> \frac{n}{2}+{c\sqrt{n\log n}} \right){\leq e^{-2c^2(n\log n)/n}=n^{-2c^2}}.
	\end{align*}
	By the union bound, the probability of finding a solution with more than $d$ 1-bits in {$n^{2c^2}/2-1$} iterations is at most \begin{equation*}
	{\sum^{{n^{2c^2}/2-1}}_{t=0} \pr{|\xi_t|_1>d}\leq n^{-2c^2}\cdot {\frac{n^{2c^2}}{2}}=\frac{1}{2}}.
	\end{equation*} 
	Because the optimal solution must have at least $(d+1)$ 1-bits, the expected running time is at least ${1/2\cdot n^{2c^2}/2}=n^{\omega(1)}$, i.e., super-polynomial.
 \end{proof}

\subsection{Deletion-robust BinVal}

For the (1+1)-EA solving deletion-robust BinVal (i.e., Eq.~\eqref{eq-deletion-binval}), Theorems~\ref{thm-del-bv-upper} and~\ref{thm-del-bv-lower} show that the tight range of $d$ allowing polynomial running time is $[1,n/2+O(\sqrt{n\log n})]$. The reason for the effectiveness of the (1+1)-EA for $d=n/2+O(\sqrt{n\log n})$ or $d\le n/2$ is similar to what has been found for deletion-robust OneMax, i.e., the (1+1)-EA can
efficiently find a solution $x$ with $d+1\le |x|_1\le k$ and then quickly find the optimum. 

%The proof of Theorem~\ref{thm-del-bv-upper} is similar to that of Theorem~\ref{del-om-upper} except the analysis of phase~(2).
\begin{theorem}\label{thm-del-bv-upper}
	If {$ d\le n/2+c\sqrt{n\log n}$, $0\le c=O(1)$}, then $\forall k>d$, the expected running time of the (1+1)-EA for deletion-robust BinVal is {${O(n^{7c^2+2}+n^2\log n)}$}, i.e., polynomial. %If $d<n/2$, the expected running time is $O(n^2\log n)$.
\end{theorem}
\begin{proof}
	Similar to the proof of Theorem~\ref{del-om-upper}, we divide the optimization process into two phases. The expected running time of phase~(1), i.e., finding a solution $x$ with $d+1\le |x|_1\le k$, is the same as that of deletion-robust OneMax, i.e., {$O(n^{7c^2+2})$}. %for $d= n/2+c$ and $O(n^2)$ for $d<n/2$. 
	We only need to analyze the expected running time of phase~(2), i.e., finding an optimal solution after phase~(1).
	
	For phase~(2), we further consider two subphases: (2a) starts after phase (1), and finishes upon finding a solution with $d+1$ leading 1-bits; (2b) starts after phase~(2a), and finishes upon finding the optimal solution $1^k0^{n-k}$.\\
	For phase~(2a), we use Theorem~\ref{multiplicative-drift} and the distance function is constructed as
	\begin{equation*}
	\begin{aligned}
	V(x)=j \quad \text{if} \; \sum_{i=1}^{d+j}x_i=d \wedge x_{d+j+1}=1.
	\end{aligned}
	\end{equation*}
	Intuitively, $V(x)$ denotes the number of 0-bits of $x$ preceding the $(d+1)$-th 1-bit. Note that $j\le n-d-1$ and $V(x)=0$ iff {$x_i=1$ for $ 1\le i\le d+1$}, i.e., the $d+1$ leading bits are all 1s. Next we examine
	$\mathrm{E}(V(\xi_t)-V(\xi_{t+1}) \mid \xi_t=x)$. Assume that currently $V(x)=j$. Note that every weight is strictly larger than the sum of all smaller weights, thus, $V(\cdot)$ will never increase. To derive a lower bound for the drift, we consider two cases: (i) $|x|_1<k$; (ii) $|x|_1=k$. Note that there are exactly $j$ 0-bits in positions 1 to $(d+j)$. For case~(i), to decrease $V(\cdot)$ by at least 1, it is sufficient that only the leftmost 0-bit of $x$ is flipped, whose probability is $1/n\cdot (1-1/n)^{n-1}\ge 1/(en)$. For case~(ii), to decrease $V(\cdot)$ by at least 1, it is sufficient that one of the leftmost $j$ 0-bits and the rightmost 1-bit of $x$ are flipped, whose probability is $j/n^2\cdot (1-1/n)^{n-2}\ge j/(en^2)$. Combining these two cases, we have $\mathrm{E}(V(\xi_t)-V(\xi_{t+1}) \mid \xi_t=x)\ge j/(en^2)=V(x)/(en^2) $. By Theorem~\ref{multiplicative-drift}, the expected running time of phase~(2a) is at most $en^2(1+\log (n-d-1))=O(n^2\log n)$.
	
	Suppose a solution $x'$ is found after phase~(2a), then the (1+1)-EA will always maintain a solution with $d+1$ leading 1-bits during phase~(2b), because any flipping of the 1-bit in positions {$1,2,\ldots,(d+1)$} will cause the fitness to decrease. Thus, the $d+1$ leading 1-bits are fixed and the fitness depends only on the remaining 1-bits, and the optimization procedure of deletion-robust BinVal is the same as that of BinVal starting from $x'$. By Theorem~13 in \cite{friedrich2018analysis}, the expected running time of the (1+1)-EA on BinVal under cardinality constraint is $O(n^2)$, thus the expected running time of phase~(2b) is also $O(n^2)$.
	
	Combing phases~(1) and~(2), the total expected running time is {${O(n^{7c^2+2}+n^2\log n)}$}, i.e., polynomial.
 \end{proof}

Applying the proof procedure of Theorem~\ref{del-om-lower}, we have:
%Theorem~\ref{thm-del-bv-lower} shows that the (1+1)-EA has to traverse a plateau where the drift is negative and the solution tends to move away from the target state.
\begin{theorem}\label{thm-del-bv-lower}	
	If {$ d= n/2+c\sqrt{n\log n}$, $c=\omega(1)$}, then $\forall k>d$, the expected running time of the (1+1)-EA for deletion-robust BinVal is at least ${n^{2c^2}/4}$, i.e., super-polynomial.
\end{theorem}

\subsection{General Cases}

For the (1+1)-EA solving deletion-robust linear optimization (i.e., Eq.~\eqref{eq-deletion-linear}), Theorems~\ref{thm-del-linear-upper} and~\ref{thm-del-linear-lower} show that the tight range of $d$ allowing polynomial running time is $d=O(1)$. The reason for the effectiveness of the (1+1)-EA when $d=O(1)$ is as follows. 
First, the (1+1)-EA can quickly find a solution $x$ with $d+1\le |x|_1\le k$. Then, we apply Theorem~\ref{multiplicative-drift}, i.e., the multiplicative drift theorem, and show that the expected decrease of the distance in each step is at least $1/(en^{2d+2})$ (i.e., $1/n^{O(1)}$) times the current distance. Thus, the expected running time can be upper bounded. Note that $\delta=\min_{w_i\neq w_j}|w_i-w_j|$ denotes the minimum difference of two different weights and $\delta:=1$ if all the weights are the same, {$x_{i:j}$ ($i<j$) denotes the substring $x_{i}x_{i+1}\ldots x_j$ of $x$.}
%for the effectiveness of the (1+1)-EA when $d=O(1)$ is as follows. Then, the solution only needs to flip at most $2d+2$, i.e., $O(1)$ bits simultaneously to make a progress towards the optima, whose probability is $1/n^{O(1)}$. Thus, the expected running time is polynomial. The proof is   The proof intuition for $d=O(1)$ is as follows. 

\begin{theorem}\label{thm-del-linear-upper}
	If $ d= O(1)$, then $\forall k>d$, $\forall \{w_i\}_{i=1}^n$, the expected running time of the (1+1)-EA for deletion-robust linear optimization is $O(n^{2d+2}\cdot (\log (kw_1)+1/\delta))$, i.e., polynomial in $n$, $w_1$ and $1/\delta$.
\end{theorem}
\begin{proof}
	Similar to the proof of Theorem~\ref{del-om-upper}, we consider two phases. First, we show that the expected running time of finding a feasible solution $x$ with $d+1\le |x|_1\le k$ (i.e., phase~(1)) is $O(n^2)$. For any solution $x$, if $|x|_1>k$, the fitness of $x$ is $k-|x|_1$ by Eq.~\eqref{def-fitness}; if $|x|_1 \leq d$, the fitness of $x$ is 0 for deletion-robust OneMax and deletion-robust linear optimization. Thus, by the analysis of phase (1) in Theorem~\ref{del-om-upper} and $d=O(1)$, the expected running time of phase~(1) is $O(n^2)$. 
	
	Then, we use Theorem~\ref{multiplicative-drift} to analyze the expected running time of phase~(2).
	 %to derive the upper bound for the expected running time.
	Let $I(x)=\{l\mid x_l=1,\sum_{i=1}^lx_i\ge d+1\}$, i.e., each element in $I(x)$ denotes the index of the $i$-th ($i\ge d+1$) 1-bit of $x$. Let $I_1(x)=I\cap \{d+1,d+2,\ldots k\}$, $I_2(x)=I\cap \{k+1,k+2,\ldots n\}$, then we have $I(x)=I_1(x)\cup I_2(x)$. The distance function is constructed as
	\begin{equation*}
	V(x)=\sum_{i=d+1}^kw_i-F(x)=\sum_{i=d+1}^kw_i-\sum_{i\in I(x)}w_i =V_1(x)-V_2(x),
	\end{equation*}
	where 
	\begin{equation*}
	V_1(x)=\sum_{d+1\le i\le k, i\notin I_1(x)}w_i,\quad V_2(x)=\sum_{i\in I_2(x)}w_i.
	\end{equation*}
	
	Note that $V(x)=0$ iff $F(x)=F(1^k0^{n-k})$, i.e., $x$ is optimal. By line~4 of Algorithm~\ref{(1+1)-EA}, $F(x)$ will not decrease, thus  $V(\cdot)$ will never increase and we only need to consider the expected decrease of $V(\cdot)$. Let $Q_x,R_x,S_x$ denote the sets of indices of 1-bits in $x_{1:d},x_{(d+1):k},x_{(k+1):n}$, respectively.  Furthermore, let $q=|Q_x|,r=|R_x|,s=|S_x|$. First we consider $q+r\le d$, i.e., $I_1=\emptyset$. We consider two cases.
	
	(a) $s\ge d-q+1$. Suppose $y$ is a solution generated from $x$ by flipping the first $d-q-r$ 0-bits of $x$ and the first $d-q-r$ 1-bits in $x_{(k+1):n}$, then we have $|y_{1:d}|_1={|Q_y|_1}=d-r$, $|y_{(d+1):k}|_1=|R_y|_1=r$ and $F(x)=F(y)$. {Note that $|Q_y|_1+|R_y|_1=d$, we have} $I_1(y)=\emptyset$ and $I_2(y)=S_y$, thus, 
	\begin{equation}\label{vy}
	V_1(y)=\sum_{i=d+1}^kw_i,\quad V_2(y)=\sum_{i\in S_y}w_i.
	\end{equation}
	We further consider two subcases.\\
	(a1) $s=k-q-r$, i.e., $x$ and $y$ both reach the cardinality bound $k$. Let ${z(j)}$ denote a solution generated as follows:
	\begin{itemize}
		\item[(M1)] flip the $r$ 0-bits  in $y_{1:d}$, the $j$-th 0-bit in $y_{(d+1):k}$ (whose index is denoted as ${{o_j}}$);
		\item[(M2)] flip the first $r$ 1-bits in $y_{(k+1):n}$ (whose indices are denoted as $S_y(1:r)$);
		\item[(M3)] flip the last 1-bit in $y_{(k+1):n}$ (whose index is denoted as $S_y(\mathrm{end})$).
	\end{itemize}
  	Then we have $|{z(j)}|_1=|x|_1=|y|_1=k$ and
	\begin{equation*}
	\begin{aligned}
	V_1({z(j)})=\sum_{d+1\le i\le k,y_i=0}w_i-w_{{{o_j}}},\quad V_2({z(j)})=\sum_{i\in S_y\backslash S_y(1:r)}w_i-w_{S_y(\mathrm{end})}.
	\end{aligned}
	\end{equation*}
	Thus, 
	\begin{align}\label{vx-vzi}
	&V(x)-V({z(j)})\\
	&=V(y)-V({z(j)})= V_1(y)-V_1({z(j)})-(V_2(y)-V_2({z(j)}))\nonumber\\
	&= \sum_{i\in R_y}w_i+w_{{{o_j}}}-\left(\sum_{i\in S_y(1:r)}w_i+w_{S_y(\mathrm{end})}\right)
	\ge w_{{{o_j}}}-w_{S_y(\mathrm{end})}, \nonumber
	\end{align}
	where the last inequality holds because $w_i$ decreases with $i$ and $|R_y|_1=r$.
	Note that $|y_{(d+1):k}|_0=k-d-r=s+q-d\ge 1$, we have 
	\begin{equation}\label{vx-vz}
	\begin{aligned}
	&\sum_{{j=1}}^{k-d-r}\left(V(x)-V({z(j)})\right)\\
	 &\ge \sum_{i\in R_y}w_i-\sum_{i\in S_y(1:r)}w_i+\sum_{d+1\le j\le k,y_i=0}\Big(w_{{{o_j}}}-w_{S_y(\mathrm{end})}\Big)\\
	& =\sum_{i=d+1}^kw_i-\left(\sum_{i\in S_y(1:r)}w_i+(k-d-r)w_{S_y(\mathrm{end})}\right)\\
	&\ge V_1(y)-V_2(y)=V(y)=V(x),
	\end{aligned}
	\end{equation}
	where the last inequality holds because $|S_y|_1=s-(d-q-r)=k-q-r-(d-q-r)=k-d$ and $w_i$ decreases with $i$. Let $E_j$ denote the event that ${z(j)}$ is generated from $x$ by bit-wise mutation, then we have $\pp(E_j)\ge (1/n)^{2d+2}\cdot (1-1/n)^{n-2d-2}{\ge 1/(en^{2d+2})} $ because at most $2(d-q-r+r)+2\le 2d+2$ bits need to be flipped. Thus, we have 
	\begin{equation}\label{E-general-case}
	\begin{aligned}
	\expect{V(\xi_t)-V(\xi_{t+1})|\xi_t=x}
	\ge \sum_{j=1}^{k-d-r}\left(V(x)-V({z(j)})\right)\cdot \pp(E_j)\ge \frac{V(x)}{en^{2d+2}}.
	\end{aligned}
	\end{equation}
	(a2) $s<k-q-r$. The analysis is similar to that of case~(a1), the main difference is that to generate ${z(j)}$, (M3) is not needed, i.e., the last 1-bit in $y_{(k+1):n}$ is not flipped. We have $|{z(j)}|_1=|y|_1+1=|x|_1+1\le k$. Then $V_2({z(j)})$ becomes $\sum_{i\in S_y\backslash S_y(1:r)}w_i$ and Eq.~\eqref{vx-vzi} becomes
	\begin{equation*}
		V(x)-V({z(j)})\ge \sum_{i\in R_y}w_i+w_{{{o_j}}}-\sum_{i\in S_y(1:r)}w_i\ge w_{{{o_j}}}.
	\end{equation*}

	Accordingly, Eq.~\eqref{vx-vz} becomes 
	\begin{equation*}
	\begin{aligned}
	\sum_{j=1}^{k-d-r}\left(V(x)-V({z(j)})\right)\ge \sum_{i=d+1}^kw_i-\sum_{i\in S_y(1:r)}w_i \ge  V_1(y)-V_2(y)=V(y)=V(x),
	\end{aligned}
	\end{equation*}
	and Eq.~\eqref{E-general-case} still holds.
	%where the last inequality is because % $|S_y|_1=s-(d-q-r)\ge r$ and $w_i$ decreases with $i$. Then, we have 
	
	(b) $s\le d-q$. Note that $|x|_1\ge d+1$, we have $d-q-r< s\le d-q$. The analysis is similar to that of case~(a1), and the main difference is %$y$ is generated from $x$ by flipping the first $d-q-r$ 0-bits in $x$ and the first $d-q-r$ 1-bits in $x_{(k+1):n}$. Then 
	${z(j)}$. Note that $|S_y|=s-(d-q-r)\le r$, %to generate ${z(j)}$,
	(M2) and (M3) becomes ``flip all the 1-bits in $y_{(k+1):n}$". Then, {$V_2(z(j))$} becomes 0 and
	%the number of 1-bits in $y_{(k+1):n}$ is . 
	Eq.~\eqref{vx-vzi} becomes 
	\begin{equation*}
		V(x)-V({z(j)})\ge \sum_{i\in R_y}w_i+w_{{{o_j}}}-\sum_{i\in S_y}w_i\ge w_{{{o_j}}},
	\end{equation*}
	and Eq.~\eqref{vx-vz} becomes 
	\begin{align*}
	\sum_{j=1}^{k-d-r}\left(V(x)-V({z(j)})\right)\ge \sum_{i=d+1}^kw_i-\sum_{i\in S_y}w_i = V_1(y)-V_2(y)=V(y)=V(x).
	\end{align*}
	Thus, Eq.~\eqref{E-general-case} still holds.
	
	For the case $q+r>d$,  we consider $z(j)$ directly instead of relying on $y$, that is, $z(j)$ is generated by mutation on $x$ instead of $y$. We consider two cases.
	
	(c) $s\ge d-q+1$. We further consider two subcases.\\
	(c1) $s=k-q-r$. Let $z(j)$ denote a solution generated from $x$ by flipping the $d-q$ 0-bits in $x_{1:d}$, the $j$-th 0-bit in $x_{(d+1):k}$, the first $d-q$ 1-bits in $x_{(k+1):n}$ and the last 1-bit in $x$. Then we have 
	\begin{equation*}
		V_1(z(j))=\sum_{d+1\le i\le k,x_i=0}w_i-w_{o_j}, V_2(z(j))=\sum_{i\in S_x\backslash S_x(1:(d-q))}w_i-w_{S_x(\mathrm{end})}.
	\end{equation*}
	Then Eq.~\eqref{vx-vzi} becomes
	\begin{equation*}
	V(x)-V(z(j))=\sum_{i\in R_x(1:(d-q))}w_i+w_{o_j}-\left(\sum_{i\in S_x(1:(d-q))}w_i+w_{S_x(\mathrm{end})}\right)\ge w_{o_j}-w_{S_x(\mathrm{end})}.
	\end{equation*}
	and Eq.~\eqref{vx-vz} becomes  
	\begin{equation*}
	\begin{aligned}
	&\sum_{j=1}^{k-d-r}(V(x)-V(z(j)))\\
	&\ge \sum_{i\in R_x(1:(d-q))}w_i-\sum_{i\in S_x(1:(d-q))}w_i+\sum_{j=1}^{k-d-r}\blueBl w_{o_j}-w_{S_x(\mathrm{end})}\blueBr\\
	&\ge \sum_{d+1\le i\le k,i\notin I_1(x)}w_i-\left(\sum_{i\in S_x(1:(d-q))}w_i+(k-d-r)w_{S_x(\mathrm{end})}\right)\\
	&\ge V_1(x)-V_2(x)=V(x).
	\end{aligned}
	\end{equation*}
	Thus, Eq.~\eqref{E-general-case} still holds.\\
	(c2) $s<k-q-r$. The analysis is similar to the analysis above, but the last 1-bit in $x$ is not flipped. Then we have 
	\begin{equation*}
	V_1(z(j))=\sum_{d+1\le i\le k,x_i=0}w_i-w_{o_j}, V_2(z(j))=\sum_{i\in S_x\backslash S_x(1:(d-q))}w_i.
	\end{equation*}
	Then Eq.~\eqref{vx-vzi} becomes
	\begin{equation*}
	V(x)-V(z(j))=\sum_{i\in R_x(1:(d-q))}w_i+w_{o_j}-\sum_{i\in S_x(1:(d-q))}w_i\ge w_{o_j}.
	\end{equation*}
	and Eq.~\eqref{vx-vz} becomes 
	\begin{equation*}
	\sum_{j=1}^{k-d-r}(V(x)-V(z(j)))\ge \sum_{d+1\le i\le k,i\notin I_1(x)}w_i-\sum_{i\in S_x(1:(d-q))}w_i\ge V_1(x)-V_2(x)= V(x),
	\end{equation*}
	Thus, Eq.~\eqref{E-general-case} still holds.
	
	(d) $s\le d-q$. Let $z(j)$ denote a solution generated from $x$ by flipping the $d-q$ 0-bits in $x_{1:d}$, the $j$-th 0-bit in $x_{(d+1):k}$, all the 1-bits in $x_{(k+1):n}$. Then we have 
	\begin{equation*}
	V_1(z(j))=\sum_{d+1\le i\le k,x_i=0}w_i-w_{o_j}, V_2(z(j))=0.
	\end{equation*}
	Eq.~\eqref{vx-vzi} becomes 
	\begin{equation*}
	V(x)-V({z(j)})\ge \sum_{i\in R_x(1:(d-q))}w_i+w_{{{o_j}}}-\sum_{i\in S_x}w_i\ge w_{{{o_j}}},
	\end{equation*}
	and Eq.~\eqref{vx-vz} becomes 
	\begin{align*}
	\sum_{j=1}^{k-d-r}\left(V(x)-V({z(j)})\right)
	\ge \sum_{i\in R_x(1:(d-q))}w_i-\sum_{i\in S_x}w_i+\sum_{j=1}^{k-d-r}w_{o_j}
	= \sum_{d+1\le i\le k,i\notin I_1(x)}^kw_i-\sum_{i\in S_x}w_i =V(x).
	\end{align*}
	Thus, Eq.~\eqref{E-general-case} still holds.
	
	Next we examine $V_{\min}$, i.e., the minimum among all possible positive values of $V$. If $w_1=w_2=\ldots=w_n$, we have $V_{\min}\ge w_n$. Otherwise, for any solution $x$ with $d+1\le |x|_1< k$, we have $\sum_{i=d+1}^kw_i-F(x)\ge w_k\ge w_n$; for any solution $x$ with $|x|_1=k$, $\sum_{i=d+1}^kw_i-F(x)\ge \delta$. Thus, we have $1/V_{\min}\le 1/w_n+1/\delta\le 1+1/\delta$. Note that $V(x)\le \sum_{i=d+1}^kw_i\le kw_1$, then by Theorem~\ref{multiplicative-drift}, the expected running time until finding an optimal solution is at most 
	\begin{equation*}
	\Big(1+\log \big(kw_1\cdot \big(1+\frac{1}{\delta}\big)\big)\Big)\cdot en^{2d+2}\le \Big(1+\log (kw_1)+\log \big(1+\frac{1}{\delta}\big)\Big)\cdot en^{2d+2}
	\le \Big(1+\log (kw_1)+\frac{1}{\delta}\Big)\cdot en^{2d+2}.
	\end{equation*}
	Combining the two phases, the total expected running time is at most $O(n^{2d+2}\cdot (\log (kw_1)+1/\delta))$, i.e., polynomial in $n$, $\log w_1$ and $1/\delta$.
 \end{proof}

Next, Theorem~\ref{thm-del-linear-lower} shows that the expected running time is super-polynomial when $d=\omega(1)$. The proof is divided into two parts based on the value of $d$. For $d=\omega(1)\cap n-\omega(1)$, we set the weights to specific values, and the (1+1)-EA needs to traverse a large plateau to find the optimum, leading to super-polynomial expected running time. For $d=n-O(1)$, the reason why the (1+1)-EA is inefficient is the same as that observed in the analysis of deletion-robust OneMax and deletion-robust BinVal. That is, the (1+1)-EA has to traverse a large plateau consisting of solutions with size at most $d$, where the drift is negative, i.e., the solution tends to move away from the target state.

\begin{theorem}\label{thm-del-linear-lower}
	If $ d= \omega(1)$, then $\exists k>d$ and $\{w_i\}_{i=1}^n$ such that the expected running time of the (1+1)-EA for deletion-robust linear optimization is super-polynomial.	
\end{theorem}
\begin{proof}
	First, we consider $d=\omega(1)\cap n-\omega(1)$. The problem is constructed as follows: $k=d+1$, {$w_i=2$ for $1\le i\le k$, $w_i=1$ for $k+1\le i\le n$}. For any $x$ with $|x|_1\le k-1$, $F(x)=0$ by Eq.~\eqref{eq-deletion-linear}; for any $x$ with $|x|_1\ge k+1$,  $F(x)=k-|x|_1<0$ by Eq.~\eqref{def-fitness}. 
	Let {$A=\{x \in \{0,1\}^n\mid |x|_1= k\}$ and $x^*=1^k0^{n-k}$, then $F(x^*)=2$ and $F(x)=1$  for any $x\in A\setminus \{x^*\}$. Thus, $x^*$ is the optimal solution.} Note that for any $x\neq x^*$, its fitness only depends on the number of 1-bits, thus, the positions of the 1-bits are treated symmetrically and the first solution with $k$ 1-bits found by the (1+1)-EA is uniformly distributed in $A$. Then, the (1+1)-EA will perform a random walk in $A$ because $g(x\in A)>g(x\notin A)$ and for any $x\in A\setminus \{x^*\}$, $g(x)=1$. Therefore, the solution is always uniformly distributed in $A$. By the union bound, the probability of finding $x^*$ in $|A|/2-1$ iterations is at most $\sum_{t=0}^{|A|/2-1}\pp(\xi_t=x^*)=|A|/2\cdot 1/|A|=1/2$.  Thus, the expected running time is at least $|A|/2\cdot 1/2=\binom{n}{k}/4$, which is super-polynomial for $d=\omega(1)\cap n-\omega(1)$.
	
	Next we consider $d= n-O(1)$, it is easy to see that $d= n/2+\omega(\sqrt{n\log n})$. By the proof of Theorem~\ref{del-om-lower}, we can derive that the expected running time until finding a solution with more than $d$ 1-bits is super-polynomial. Note that an optimal solution must have $k>d$ 1-bits, thus, the expected running time is super-polynomial.
 \end{proof}

\section{Worst-case Linear Optimization}\label{sec-worst}

In this section, we consider the (1+1)-EA for worst-case linear optimization (i.e., Eq.~\eqref{eq-worst-linear}). Theorem~\ref{thm-wc-upper} shows that when $k=O(1)$ or $k=n-O(1)$, the expected running time is polynomial. 
For $k=O(1)$, the Hamming distance between a feasible solution and an optimal solution is at most $2k$, i.e., $O(1)$, thus, the (1+1)-EA can quickly jump to the optima. 
For $k=n-O(1)$, if the size of a solution $x$ is exactly $k$, 
then the Hamming distance between $x$ and an optimal solution is at most $2(n-k)$, i.e., $O(1)$; if $|x|_1< k$, then $x$ can be improved by flipping its 0-bits. Thus, the (1+1)-EA can also efficiently find the global optima.
The proof is accomplished by applying Theorem~\ref{additive-drift}, i.e., the additive drift theorem. Note that $w_{\max}$ denotes the maximum weight of all $m$ linear functions.
\begin{theorem}\label{thm-wc-upper}
	If $ k= O(1)$, then $\forall m \geq 1$, $\forall \{f_s\}_{s=1}^m$, the expected running time of the (1+1)-EA for worst-case linear optimization is $O(n^{2k})$, i.e., polynomial.\\
	If $ k= n-O(1)$, then $\forall m \geq 1$, $\forall \{f_s\}_{s=1}^m$, the expected running time of the (1+1)-EA for worst-case linear optimization is $O(n^{2(n-k+1)} \cdot w_{\max})$, i.e., polynomial in $n$ and $w_{\max}$.
\end{theorem}
\begin{proof}
	From the proof of Theorem~\ref{del-om-upper}, the expected running time of finding a feasible solution is at most $O(n\log n)$. Then we consider the expected running time until an optimal solution is found. Note that the optimal solutions can be non-unique, and we only need to find one optimal solution $x^*$.	First we consider $k=O(1)$. For any feasible solution $x$, we have
	\begin{equation*}
	H(x,x^*)=\sum_{i=1}^n|x_i-x^*_i|\le \sum_{i=1}^n(|x_i|+|x^*_i|)\le 2k,
	\end{equation*}
	where $H(x,x^*)$ denotes the Hamming distance between $x$ and $x^*$, and the last inequality holds because $x$ and $x^*$ are both feasible, i.e., $|x|_1\le k\wedge |x^*|_1\le k$. Thus, it requires to flip at most $2k$ bits of $x$ to generate $x^*$, whose probability is at least $1/n^{2k}\cdot (1-1/n)^{n-2k}\ge 1/(en^{2k})$ . This implies that the expected running time until finding an optimal solution is at most $O(n^{2k})$ . Combining the two phases, the total expected running time is at most $O(n^{2k})$, i.e., polynomial.
	
	Then we consider $k=n-O(1)$. We use Theorem~\ref{additive-drift} to derive the expected running time until finding an optimal solution. Let $M=F(x^*)$, i.e., $M$ is the maximum objective value of the problem. The distance function $V(x)$ is constructed as
	\begin{align*}
	V(x)=\begin{cases}
	M+1-\lfloor F(x) \rfloor & F(x)<M,\\
	0 & F(x)=M.
	\end{cases}
	\end{align*}
	Thus, $V(x=0)$ if and only if $x$ is an optimal solution, i.e., $x \in \mathcal{X}^*$. Then, we examine $\expect{V(\xi_t)-V(\xi_{t+1}) \mid \xi_t=x}$ for any $x$ with $F(x)<M$. $V(\cdot )$ will never increase and we only need to consider the expected decrease of $V(\cdot)$. \\
	If $|x|_1\le k-1$, to decrease $V(\cdot)$, it is sufficient that exactly one 0-bit of $x$ is flipped, whose probability is $(n-|x|_1)/n\cdot (1-1/n)^{n-1}\ge 1/(en)$. Note that ${w_{s,i}}\ge 1$, we have $\expect{V(\xi_t)-V(\xi_{t+1}) \mid \xi_t=x}\ge 1\cdot 1/(en)=1/(en)$.\\
	If $|x|_1=k$, we consider that one optimal solution $x^*$ is generated from $x$. Note that
	\begin{equation*}
	H(x,x^*)\!=\!\sum_{i=1}^n|x_i-1+1-x^*_i|\!\le\! \sum_{i=1}^n(|x_i-1|+|x^*_i-1|)\!=\!2(n-k),
	\end{equation*}
	thus $x^*$ will be generated with probability at least $1/n^{2(n-k)}\cdot (1-1/n)^{2k-n}\ge 1/(en^{2(n-k)})$ in one iteration. As $V(x)-V(x^*)=V(x)\ge 1$, we have $\expect{V(\xi_t)-V(\xi_{t+1}) \mid \xi_t=x}\!\ge\! 1/(en^{2(n-k)})$. \\
	Combining the two cases, we have $\expect{V(\xi_t)-V(\xi_{t+1}) \mid \xi_t=x}\ge 1/(en^{2(n-k)+1})$. By Theorem~\ref{additive-drift}, the expected running time until finding an optimal solution is at most
	\begin{equation*}
	\begin{aligned}
	(M+1)\cdot en^{2(n-k)+1}\le e(kw_{\max}+1)n^{2(n-k)+1}	=O\blue{\big(} n^{2(n-k+1)}\cdot w_{\max}\blue{\big)},
	\end{aligned}
	\end{equation*}
	where $w_{\max}=\max_{1\le i\le n,1\le s\le m}{w_{s,i}}$. Thus, the total expected running time is at most $O(n\log n)+O(n^{2(n-k+1)}\cdot w_{\max})$ $=O(n^{2(n-k+1)}\cdot w_{\max})$, i.e., polynomial in $n$ and $w_{\max}$.
 \end{proof}

In the following theorem, we show that the expected running time is super-polynomial when $k=\omega(1)\cap n-\omega(1)$. Moreover, if $k=O(1)$, the theorem gives a lower bound of $\Omega(n^{2k})$, matching the general upper bound in Theorem~\ref{thm-wc-upper}. We prove Theorem~\ref{thm-wc-lower} by constructing different objective functions $\{f_s\}_{s=1}^m$ for different values of $k$, and the proof intuitions  are also different. For $k<n/2$, the (1+1)-EA can easily get stuck in local optima, and needs to flip $2k$ bits simultaneously to escape from local optima, whose probability is at most $1/n^{2k}$. For $k\ge n/2 \wedge k=n-\omega(1)$, the reason why the (1+1)-EA is inefficient is similar to that for $d=\omega(1)\cap n-\omega(1)$ in Theorem~\ref{thm-del-linear-lower}, i.e., the (1+1)-EA needs to traverse a large plateau to find the optimum.
\begin{theorem}\label{thm-wc-lower}
	If $ k< n/2$, then $\forall m\ge 2$,  $\exists \{f_s\}_{s=1}^m$ such that the expected running time of the (1+1)-EA for worst-case linear optimization is at least $(n/4)^{2k}$. \\
	If $k\ge n/2\wedge k= n-\omega(1)$, then $\exists m$ and $\{f_s\}_{s=1}^m$ such that the expected running time of the (1+1)-EA for worst-case linear optimization is at least $\binom{n}{k}/4$, i.e., super-polynomial.
\end{theorem}
\begin{proof}
	First, we consider $k=1$, i.e., any feasible solution has exactly one 1-bit. The objective functions $\{f_s\}_{s=1}^m$ are constructed as follows: $\forall 1\le s\le m, 2\le i\le n: w_{{s,1}}=2,w_{{s,i}}=1$. {Let $x^*=10^{n-1}$, then $F(x^*)=2$; for any $x\neq x^*$, $F(x)\le 1$, thus $x^*$ is the optimal solution.}
	 For any $|x|_1>1$, we have $g(x)=k-|x|_1$; for any $|x|_1=0$, we have $g(x)=0$. Thus, the positions of the 1-bits are treated symmetrically until finding a solution $y$ with $|y|_1=1$, and $y\neq x^*$ with probability $1-1/n$. Note that for $y\neq x^*$, it needs to flip the 1-bit and the leftmost 0-bit to generate $x^*$, whose probability is $1/n^2\cdot (1-1/n)^{n-2}\le 1/n^2$. Thus, the expected running time is at least $(1-1/n)\cdot n^2\ge (n/4)^2$.
	
	Next we consider $2\le k< n/2$. The functions $\{f_s\}_{s=1}^m$ are constructed as
	\begin{equation*}
	\begin{aligned}
	\forall 1\le s\le m-1: \ \;	& {w_{s,i}=k+1 \textrm{ for }1\le i\le k-1} , &\!\! {w_{s,k}}=3/2,\:\; & {w_{s,i}=k \textrm{ for }k+1\le i\le n}, \\
	& {w_{m,i}=1 \textrm{ for }1\le i\le k-1} , & \!\! {w_{m,k}}=k^2,\:\; & {w_{m,i}=k \textrm{ for }k+1\le i\le n},
	\end{aligned}
	\end{equation*}
	We will show that such a problem has one global optimal solution $x^*=1^k0^{n-k}$ and a set of local optima $A=\{x\mid |x|_1=k,x_{1:k}=0\}$. For $x^*$, we have $\forall s\le m-1$, $f_s(x^*)=(k+1)(k-1)+3/2=k^2+1/2$ and $f_m(x^*)=k-1+k^2$. Thus,   $F(x^*)=k^2+1/2$.
	For any $x\in A$, we have $\forall s\le m, f_s(x)=k^2$, thus, $F(x)=k^2$.
	For any $x\notin \{x^*\}\cup A$ , we consider two cases:\\
	(1) $|x|_1\le k-1$, it can be verified that $F(x)<k^2$  because the weight of each element is at most $k+1$. \\
	(2) $|x|_1=k$,  let $j=\sum_{i=1}^kx_i$, then we have $1\le j\le k-1$. We further consider two subcases.\\
	(2a) $x_k=0$. We have $\forall s\le m-1$, $f_s(x)= j(k+1)+(k-j)k=k^2+j$ and $f_m(x)=j+k(k-j)=k^2-kj+j\le k^2-j$. Thus,  $F(x)=f_m(x)< k^2$. \\
	(2b) $x_k=1$. We have $\forall s\le m-1$, $f_s(x)= (j-1)(k+1)+3/2+(k-j)k=k^2-k+j+1/2\le k^2-1/2$ and $f_m(x)\ge k^2$.	Thus,  $F(x)\le k^2-1/2$.\\
	Combining the two cases, we have 
	\begin{equation*}
	F(x\notin \{x^*\}\cup A)<F(x\in A)<F(x^*).
	\end{equation*}
	Next we examine the probability that the (1+1)-EA finds a solution $x\in A$. For the initial solution $x$, it falls into the infeasible region with probability at least $1/2$ by the uniform initial distribution. The (1+1)-EA will minimize the number of 1-bits of a solution until finding a feasible solution $y$. Note that in this procedure, the positions of the 1-bits are treated symmetrically. Next we bound the probability that $y\in A$. For any solution with $|x|_1>k$, we have
	\begin{equation*}
	\begin{aligned}
	\pmut(x,{y^{(k)}})\ge \frac{\binom{|x|_1}{|x|_1-k}}{n^{|x|_1-k}}\left(1-\frac{1}{n}\right)^{n-|x|_1+k}\ge \frac{\binom{|x|_1}{|x|_1-k}}{en^{|x|_1-k}},
	\end{aligned}
	\end{equation*}
	where ${y^{(k)}}$ denotes any solution with exactly $k$ 1-bits and $\pmut(x,{y^{(k)}})$ denotes the probability that ${y^{(k)}}$ is generated from  $x$  by bit-wise mutation.  Meanwhile, we have
	\begin{equation*}
	\pmut(x,{y^{(<k)}})\le \frac{\binom{|x|_1}{|x|_1-k+1}}{n^{|x|_1-k+1}},
	\end{equation*}
	where ${y^{(<k)}}$ denotes any solution with less than $k$ 1-bits and the inequality holds because at least $(|x|_1-k+1)$ 1-bits needs to be flipped. Thus,
	\begin{equation*}
	\begin{aligned}
	\frac{\pmut(x,{y^{(k)}})}{\pmut(x,{y^{(<k)}})}\ge \frac{n\binom{|x|_1}{|x|_1-k}}{e\binom{|x|_1}{|x|_1-k+1}}
	& =\frac{n(|x|_1-k+1)}{ek}\ge \frac{2n}{en/2}\ge 1,
	\end{aligned}
	\end{equation*}
	where the second inequality holds by $k< n/2$. Thus,  conditional on the event that $|y|_1\le k$, we have
	\begin{equation*}
	\pp(|y|_1=k\mid|y|_1\le k)\ge 1/2.
	\end{equation*}
	Recall that the positions of the 1-bits are treated symmetrically under the condition that $|y|_1=k$, we have $\forall 2\le k<n/2$,
	\begin{equation*}
	\pp(y\in A\mid |y|_1=k)=\frac{\binom{n-k}{k}}{\binom{n}{k}} 
	\ge \frac{((n-k)/k)^k}{(en/k)^k} =\frac{1}{e^k}\cdot \blueBl 1-\frac{k}{n}\blue{\Big)^k} >\blueBl\frac{1}{2e}\blue{\Big)^k},
	\end{equation*}
	where the last inequality is by $k<n/2$.
	Thus,  we have
	\begin{equation*}
	\begin{aligned}
	\pp(y\in A\mid |y|_1\le k)=\pp(y\in A\mid |y|_1=k)\pp(|y|_1=k\mid |y|_1\le k) \ge {\frac{1}{2(2e)^{k}}},
	\end{aligned}
	\end{equation*}
	i.e., starting from an infeasible solution, the (1+1)-EA will find a solution in $A$ with probability at least {$1/(2\cdot(2e)^k)$}. Note that the initial solution is infeasible with probability at least $1/2$, the (1+1)-EA finds a solution in $A$ with probability at least {$1/(4\cdot (2e)^{k})=1/(4^{1/k}2e)^{k}\ge 1/(4e)^{k}\ge 1/4^{2k}$}. Once the (1+1)-EA finds a solution $x\in A$, it will stay in $A$ or jump to $x^*$ with probability $1/n^{2k}\cdot (1-1/n)^{n-2k}\le 1/n^{2k}$.  Thus, the expected running time is at least $n^{2k}/4^{2k}\ge (n/4)^{2k} $.

	Finally, we examine $k\ge n/2\wedge k=n-\omega(1)$. We consider $m=k$ and the functions $\{f_s\}_{s=1}^k$ are constructed as
	\begin{equation*}
	\forall 1\le s\le k:\; {w_{s,s}}=n,{w_{s,i\neq s}}=1.
	\end{equation*}
	Let $A=\{x\mid |x|_1= k\}$ and $x^*=1^k0^{n-k}$. Then $F(x^*)=n+k-1$, and for any $x$ with $|x|_1< k$ or $x\in A\setminus \{x^*\}$, there exists $1\le i\le k$ such that $x_i=0$, which implies $F(x)=f_i(x)=|x|_1$. Thus, $x^*$ is the optimal solution. For any $|x|_1>k$, we have $g(x)=k-|x|_1$.  Thus, the positions of the 1-bits are treated symmetrically and $g(x\in A)>g(x\notin A)$, $g(x)=k$ for any $x\in A\setminus \{x^*\}$. Similar to the proof of Theorem~\ref{thm-del-linear-lower}, the expected running time is at least $|A|/2\cdot 1/2=\binom{n}{k}/4=n^{\omega(1)}/4$, i.e., super-polynomial.
 \end{proof}

\section{Conclusion and Discussion}\label{sec-conclusion}

In this paper, we analyze the running time of the (1+1)-EA for robust linear optimization with a cardinality constraint $k$, including two common robust settings, i.e., deletion-robust and worst-case. Tight bounds on $d$ (i.e., the maximum number of 1-bits that can be deleted) or budget $k$ for the (1+1)-EA to solve each concerned problem in polynomial running time are derived, showing the potential of EAs for robust optimization. Note that this work is only a first step towards theoretically analyzing EAs for robust optimization. We consider relatively simple functions and assume that the objective $F$ can be obtained exactly. In practice, $F$ can be approximated by taking the minimum over a number of randomly sampled disturbances for the deletion-robust scenario, or randomly sampled objectives for the worst-case scenario. The number of called objective function evaluations will influence the approximation quality, and its relationship with the overall performance of algorithms deserves to be studied in the future. It is also interesting to study more complicated EAs on more general robust optimization problems, e.g., robust submodular optimization where the objective function is only required to satisfy the submodular property.

The deletion-robust optimization looks similar to the optimization under prior noise~\cite{giessen2014robustness,qian2018noise,droste2004analysis}, which flips some bits of a solution before fitness evaluation. When evaluating the fitness of a solution in these two types of optimization tasks, the solution is both disturbed and the objective function is affected. However, their goals are quite different. For robust optimization, the goal is to find a solution robust against disturbance, and thus, the disturbed objective is the true one of the optimization; but for noisy optimization, the goal is still to find an optimal solution with respect to the original objective function, rather than the noisy one.

The worst-case optimization can also be connected to multi-objective optimization, because they both need to consider several objectives simultaneously. For multi-objective optimization, we usually need to find a set of solutions (i.e., Pareto optimal solutions) to trade off different objectives. A single Pareto optimal solution may perform very well on some objectives but terribly on some other objectives, or just performs equally well on all objectives. Thus, the set of all Pareto optimal solutions takes into account different requirements from different users. For worst-case optimization, the trade-off is relatively simple, because we only need to consider the worst performance of all objectives. Thus, we can probably refer to multi-objective EAs to design efficient algorithms for worst-case optimization.

\section*{Acknowledgements}
The authors want to thank the editor and anonymous reviewers for their helpful comments and suggestions. This work was supported by the National Science Foundation of China (61876077, 61672478) and Fundamental Research Funds for the Central Universities (14380004).

\bibliographystyle{elsarticle-num}
\bibliography{ectheory,robust}

\end{document}